\documentclass[a4paper,11pt]{article}

\usepackage{amsthm,amssymb,amsmath}

\newtheorem{theorem}{Theorem}

\newtheorem{proposition}[theorem]{Proposition}

\newtheorem{corollary}[theorem]{Corollary}
\newtheorem{definition}[theorem]{Definition}

\newcommand{\pNP}{{\normalfont\emph{para}-\NP{}}}
\newcommand{\XP}{{\normalfont XP}}
\newcommand{\W}[1][XXX]{{\normalfont W[#1]}}
\newcommand{\M}[1][XXX]{{\normalfont M[#1]}}
\newcommand{\A}[1][XXX]{{\normalfont A[#1]}}
\newcommand{\AW}[1][XXX]{{\normalfont AW[#1]}}
\newcommand{\EW}[1][XXX]{{\normalfont EW[#1]}}
\newcommand{\EXPW}[1][XXX]{{\normalfont EXPW[#1]}}
\newcommand{\SW}[1][XXX]{{\normalfont S[#1]}}
\newcommand{\NP}{{\normalfont NP}}
\newcommand{\Poly}{{\normalfont P}}
\newcommand{\FPT}{{\normalfont FPT}}
\newcommand{\EPT}{{\normalfont EPT}}
\newcommand{\EXPT}{{\normalfont EXPT}}
\newcommand{\SUBEPT}{{\normalfont SUBEPT}}

\newcommand{\pcp}[2]{{\rm PCP[$#1$, $#2$]}}
\newcommand{\ppcp}[2]{{\rm\emph{p}-PCP[$#1$, $#2$]}}
\newcommand{\PSPACE}{{\normalfont PSPACE}}
\newcommand{\NEXP}{{\normalfont NEXP}}
\newcommand{\WSAT}[1]{\textsc{WSAT(#1)}}
\newcommand{\AWSATL}[2]{\textsc{AWSAT$_{#1}$(#2)}}
\newcommand{\AWSAT}[1]{\textsc{AWSAT(#1)}}
\newcommand{\WSATTN}{\WSAT{2-CNF$^{-}$}}

\newcommand{\Card}[1]{|#1|}
\newcommand{\Yes}{\textsc{Yes}}
\newcommand{\No}{\textsc{No}}
\newcommand{\SB}{\{\,}
\newcommand{\SM}{\;{|}\;}
\newcommand{\SE}{\,\}}

\newcommand{\pcproblem}[4]{\begin{samepage}\begin{quote} \textsc{#1}\\ 
\textit{Instance:} #2\\ \textit{Parameter:} #3\\ \textit{Question:} #4 \end{quote}\end{samepage}}

\DeclareMathOperator{\var}{var}
\DeclareMathOperator{\cl}{cl}
\DeclareMathOperator{\arity}{arity}

\begin{document}

\title{A Proof Checking View of Parameterized Complexity}
\author{Luke Mathieson}
\date{}
\maketitle

\begin{abstract}
The PCP Theorem is one of the most stunning results in computational complexity theory, a culmination of a series of results regarding proof checking it exposes some deep structure of computational problems. As a surprising side-effect, it also gives strong non-approximability results. In this paper we initiate the study of proof checking within the scope of Parameterized Complexity. In particular we adapt and extend the \pcp{n\log\log n}{n\log\log n} result of Feige \emph{et al.} to several parameterized classes, and discuss some corollaries.
\end{abstract}

\section{Introduction}

The straight-forward view of most computational complexity classes is one of what problems can solved given certain computing power and resource restrictions. Alongside this is the \emph{verification} view of complexity, where we ask not what can be computed within a given set of restrictions, but whether a given solution can be verified under certain restrictions. The most famous of these is of course the equivalent definitions of $\NP{}$ as the class of all problems that can be solved in nondeterministic polynomial time or verified in deterministic polynomial time. This definition may be thought of as a \emph{proof system}, where a Turing Machine (the verifier) has access to the input and a proof, and in polynomial time checks that the proof is correct.

With access to a random bit string, it is possible to reduce the number of bits that the verifier reads from the proof. In fact, in the case of \NP{}, this is quite a surprising reduction; with only a logarithmic number of random bits, we need only a \emph{constant} number of bits from the proof to verify the proof. The trade-off being that if the proof is false, we may incorrectly accept it, but with probability at most one half.

Such proof systems have been well studied for traditional complexity classes such as \NP{}, \PSPACE{} and \NEXP{}. In this paper we begin to look at parameterized complexity through the same lens. In particular we demonstrate a relatively simple but non-trivial proof system for $\W[1]$. We also extend this to $\W[2]$, $\M[1]$, the bounded classes $\EW[1]$, $\EXPW[1]$ \& $\SW[1]$ and the classes of the $A$-hierarchy up to $\AW[*]$.

\subsection{Useful History}

This idea of classifying languages by membership proofs began to attract serious attention in the early to mid eighties, with Goldwasser, Micali \& Rackoff's~\cite{GoldwasserMR85} introduction of the idea of \emph{interactive proofs} (later published in a more complete form~\cite{GoldwasserMR89}) and Babai's~\cite{Babai85, BabaiM88} \emph{Arthur-Merlin games}. Both probabilistic approaches to proof verification.

Over time these classes were linked back to traditionally defined complexity classes. The class of problems with interactive proofs is precisely \PSPACE{}~\cite{Shamir90}. The class of problems with Arthur-Merlin style verifiers that use a polynomial number of rounds turns out to be the same as the class of problems with interactive proofs~\cite{GoldwasserS86}. If multiple, non-communicating provers (defined in~\cite{Ben-OrGKW88}) are allowed we obtain \NEXP{}~\cite{BabaiFL91, BabaiFL92} (Ben-Or \emph{et al.}~\cite{Ben-OrGKW88} also showed that for any number of provers, there was an equivalent protocol with at most two provers).

This work culminated in the development of \emph{probabilistically checkable proofs}~\cite{AroraS92} and what is now known as the PCP Theorem:

\begin{theorem}[The PCP Theorem~\cite{AroraLMSS98, AroraS98}]\label{thm:pcp}
$\NP{}$ is the class of all languages that can be verified by a polynomial-time probabilistic Turing Machine (the verifier) that can access at most $O(\log n)$ random bits and at most $O(1)$ bits of an oracle string (the proof) such that any input that is in the language is accepted with probability $1$ and any input that is not in the language is accepted with probability at most $\frac{1}{2}$.
\end{theorem}

Dinur~\cite{Dinur07} gives more accessible proof, via constraint satisfaction.

Far from being a theoretical curiosity, PCPs have a number of applications across computer assisted mathematics~\cite{BabaiFLS91} and cryptology~\cite{GoldwasserMR89} but possibly most interestingly PCP results have implications for approximation algorithms. It is PCP results that led to inapproximability results for \textsc{Max-Word}~\cite{Condon91}, \textsc{Max-3SAT}~\cite{AroraLMSS98}, \textsc{Max-Clique}~\cite{FeigeGLSS96} and in general that if $\text{\Poly{}} \neq \text{\NP{}}$ then no MAXSNP-hard problem is in PTAS.

\section{Parameterized Complexity Theory}

A \emph{parameterized problem} is a decision problem augmented with a special input, the \emph{parameter}. This may be more formally viewed as a language over some alphabet with a \emph{parameterization} that provides a positive integer parameter for each instance. 

\begin{definition}[Parameterized Problem]
A parameterized problem over alphabet $\Sigma$ is a pair $(\Pi,\kappa)$ where $\Pi \subseteq \Sigma^{*}$ and $\kappa:\Sigma^{*}\rightarrow\mathbb{N}$ is a parameterization.
\end{definition}

Typically given an instance, the parameterization (as a function) is implied and we treat inputs as being accompanied by a integer, usually denoted $k$. 

Parameterization allows a more relaxed notion of tractability:

\begin{definition}[Fixed-parameter Tractability]
A parameterized problem $(\Pi,\kappa)$ is fixed-parameter tractable if there is an algorithm $\mathcal{A}$ and a computable function $f$ such that for all inputs $(x,\kappa(x))$ the algorithm $\mathcal{A}$ decides if $x\in\Pi$ in time bounded by $f(\kappa(x))\cdot{}\Card{x}^{O(1)}$. The class of all fixed-parameter tractable problems is \FPT{}. 
\end{definition}

This then gives a natural reduction schema:

\begin{definition}[FPT Reductions]
Given two parameterized problems $(\Pi_{1},\kappa_{1})$ over $\Sigma_{1}$ and $(\Pi_{2},\kappa_{2})$ over $\Sigma_{2}$, an fpt reduction from $(\Pi_{1},\kappa_{1})$ to $(\Pi_{2},\kappa_{2})$ is a mapping $R:\Sigma_{1}^{*}\rightarrow\Sigma_{2}^{*}$ such that for all $x\in\Sigma_{1}^{*}$:
\begin{enumerate}
\item $x \in \Pi_{1} \Leftrightarrow R(x) \in \Pi_{2}$.
\item $R$ can be computed in time bounded by $f(\kappa(x))\cdot{}\Card{x}^{O(1)}$.
\item There is a computable function $g$ such that $\kappa_{2}(R(x)) \leq g(\kappa_{1}(x))$.
\end{enumerate}
\end{definition}

The last condition results in a very rich intractability theory for parameterized complexity. We will give details of the classes relevant for this paper, but a much fuller treatment can be found in the monographs of Downey \& Fellows~\cite{DowneyFellows99} and Flum \& Grohe~\cite{FlumGrohe06}.

We first define a hierarchy of propositional logic formul\ae{}. Let $\{a_{i}\}$ be a set of boolean literals, then we define the following formula classes:
\begin{eqnarray*}
\Gamma_{0,d} := \SB a_{1} \wedge \ldots \wedge a_{c} \SM c \leq d \SE\\
\Delta_{0,d} := \SB a_{1} \vee \ldots \vee a_{c} \SM c \leq d \SE
\end{eqnarray*}

These can then be recursively stacked to give the classes $\Gamma_{t,d}$ and $\Delta_{t,d}$:

\begin{eqnarray*}
\Gamma_{t,d} := \SB \bigwedge_{i \in I} \phi_{i} \SM \phi_{i} \in \Delta_{t-1,d}  \SE\\
\Delta_{t,d} := \SB \bigvee_{i \in I} \phi_{i} \SM \phi_{i} \in \Gamma_{t-1,d}  \SE
\end{eqnarray*}

In addition we denote by $\Phi^{+}$ the subclass of a class of propositional formul\ae{} $\Phi$ where no literals are negated and by $\Phi^{-}$ the subclass of $\Phi$ where all literals are negated. Given a propositional formula over a variable set $X$ a truth assignment that sets $k$ variables of $X$ to \textsc{TRUE} is called a \emph{weight $k$ assignment}\footnote{This use of ``weight'' is standard in the parameterized complexity literature, but may conflict with definitions from other areas. In this paper, when we refer to the weight of an assignment, this is the meaning we intend.} or an assignment of weight $k$.

The fundamental problem for many parameterized intractability classes is the \textsc{Weighted Satisfiability} problem:

\pcproblem{WSAT($\Phi$)}{A boolean formula $\phi \in \Phi$ and a positive integer $k$.}{$k$.}{Is there a satisfying assignment for $\phi$ of weight $k$?}

We can then define the $W$-hierarchy:

\[\W[t] = \left[\WSAT{$\Gamma_{t,d}$}\right]^{\FPT{}}\]

where $t+d > 2$ and $[X]^{\FPT{}}$ denotes the closure of a parameterized problem $X$ under fpt reductions.

Even though we do not have quite the latitude to reduce the structure of the formula as in classical complexity (where everything in \NP{} can be reduced to a formula in 3-CNF), we can impose slightly more restriction to the formul\ae{}. In particular:

\[\W[1] = \left[\WSAT{$\Gamma_{1,2}^{-}$}\right]^{\FPT{}}\]

and

\[\W[2] = \left[\WSAT{$\Gamma_{2,1}^{+}$}\right]^{\FPT{}}\]

So for every problem in $\W[1]$ we can convert any instance into an instance of the \textsc{Weighted Satisfiability} problem where the formula is in 2-CNF and all literals are negated and for every problem in $\W[2]$ we can convert any instance into a CNF formula (of unbounded clause length) where all literals are positive (similar statements can be made for the other classes in the $W$-hierarchy, \emph{q.v.}~\cite{FlumGrohe06}).

At the other end of the parameterized intractability scale is the direct definitional analog of \NP{}:

\begin{definition}[\pNP{}]
A parameterized problem $(\Pi, \kappa)$ is in \pNP{} if there is a computable function $f$ and nondeterministic Turing Machine that on input $(x,\kappa(x))$ decides $x\in\Pi$ in time bounded by $f(\kappa(x))\cdot\Card{x}^{O(1)}$.
\end{definition}

It turns out however that \pNP{}-complete problems seem much harder than $\W[1]$-complete problems and that $\W[1]$ provides a more natural analog of \NP{}\footnote{Very loosely speaking, barring a collapse, \pNP{}-complete problems correspond to problems with time complexity $(\kappa(x))^{\Card{x}}$ or worse, whereas $\W[1]$-complete problems have complexity $\Card{x}^{\kappa(x)}$ (this bound is more formal than the given \pNP{} one as the $W$-hierarchy is contained in \XP{}~\cite{FlumGrohe06}.)}.

The class \XP{} provides an alternate perspective on parameterized intractability:

\begin{definition}
A parameterized problem $(\Pi, \kappa)$ is in \XP{} if there exists a computable function $f$ such that every instance $(x,\kappa(x))$ is decidable in time
\[
\Card{x}^{f(\kappa(x))}+f(\kappa(x))
\]
\end{definition} 

The entirety of the $W$-hierarchy is contained in \pNP{} $\cap$ \XP{}.

\XP{} in a certain sense plays a role similar to a parameterized version of EXPTIME, and as such contains a hierarchy that bears a relationship to the polynomial hierarchy and \PSPACE{}, the $A$-hierarchy.

Similar to the polynomial hierarchy, the $A$-hierarchy can be characterized by alternating quantified satisfiability problems. In this case of course, there is a parameterized flavour:

\pcproblem{AWSAT$_{l}$($\Phi$)}{A boolean propositional formula $\phi \in \Phi$, with the variable set $X$ partitioned into $l$ sets $X_{1}, \ldots, X_{l}$ and positive integers $k_{1}, \ldots, k_{l}$.}{$k = \sum_{i\in [l]}k_{i}$.}{Is there a $k_{1}$-sized subset of $X_{1}$ such that for all $k_{2}$-sized sets of $X_{2}$ there exists a $k_{3}$-sized subset of $X_{3}$... (\&c. for $l$ alternations) such that setting those variables to true satisfies $\phi$?}

If we employ the notation $\forall_{k}$ and $\exists_{k}$ to denote ``for all $k$-sized subsets'' and ``there exists a $k$-sized subset'' respectively, we can reframe the slightly awkward definition of \textsc{AWSAT$_{l}$} by asking if
\[
\exists_{k_{1}}X_{1}\forall_{k_{2}}X_{2}\ldots Q_{k_{l}}X_{l}\phi
\]
is true, where $Q\in\{\forall,\exists\}$. If we remove the bound on $l$, then we obtain the \textsc{AWSAT} problem, which has the same essential structure. When talking about this family of problems informally, we will omit the subscript and refer to them generally as \textsc{AWSAT} problems. These classes then provide the basis for the $A$-hierarchy:
\[
\A[l] = \left\{\begin{array}{ll}
[\AWSATL{l}{$\Gamma_{1,2}^{-}$}]^{\FPT{}} & \text{for $l$ odd}\\
{}[\AWSATL{l}{$\Delta_{1,2}^{+}$}]^{\FPT{}}& \text{for $l$ even}\\
\end{array}\right.
\]

One interesting superclass of the $A$-hierarchy is $\AW[*]$:
\[
\AW[*] = [\AWSAT{$\Gamma_{1,2}^{-}$}]^{\FPT{}}
\]
Thus $\AW[*]$ is not entirely dissimilar to \PSPACE{}\footnote{Or something between \PSPACE{} and PH, though this also imprecise as natural parameterized versions of some \PSPACE{}-complete problems are $\AW[*]$-complete. Conversely \AWSAT{PROP}, the parameterized alternating satisfiability problem for the class of all propositional formul\ae{}, is $\AW[SAT]$-complete and $\AW[*] \subseteq \AW[SAT]$.}, however in the parameterized setting, there is no single analog of \PSPACE{}, with its role being spread between $\AW[*]$, $\AW[SAT]$, $\AW[P]$, XL and para-\PSPACE{}~\cite{FlumGrohe06}.

\subsection{Bounded Parameterized Complexity Classes}

In the definition of \FPT{} the function $f$ that gives the dependence on the parameter is only restricted to being computable. We can define analogs of \FPT{} and its intractability hierarchies with stronger restrictions on $F$ that still retain very similar structures.

\begin{definition}[\EXPT{}]
A parameterized problem $(\Pi,\kappa)$ is in \EXPT{} if there is an algorithm $\mathcal{A}$ and such that for all inputs $(x,\kappa(x))$ the algorithm $\mathcal{A}$ decides if $x\in\Pi$ in time bounded by $2^{\kappa(x)^{O(1)}}\cdot{}\Card{x}^{O(1)}$.
\end{definition}

\begin{definition}[\EPT{}]
A parameterized problem $(\Pi,\kappa)$ is in \EPT{} if there is an algorithm $\mathcal{A}$ and such that for all inputs $(x,\kappa(x))$ the algorithm $\mathcal{A}$ decides if $x\in\Pi$ in time bounded by $2^{O(\kappa(x))}\cdot{}\Card{x}^{O(1)}$.
\end{definition}

\begin{definition}[\SUBEPT{}]
A parameterized problem $(\Pi,\kappa)$ is in \SUBEPT{} if there is an algorithm $\mathcal{A}$ and such that for all inputs $(x,\kappa(x))$ the algorithm $\mathcal{A}$ decides if $x\in\Pi$ in time bounded by\footnotemark\ $2^{o^{eff}(\kappa(x))}\cdot{}\Card{x}^{O(1)}$.
\end{definition}

\footnotetext{$f\in o^{eff}(g)$ if there exists a computable, nondecreasing, unbounded function $h:\mathbb{N}\rightarrow\mathbb{N}$ such that $f(k) \leq \frac{g(k)}{h(k)}$.}

Typically the parameterizations of problems in \SUBEPT{} are of a different character to normal parameterizations. In the subexponential theory the parameterizations play the role of ``size measures'' for the problem, rather than being independent of the size of the problem. Such measures may be for example the number of variables in a logic sentence or the number of edges and vertices in a graph (this is also in contrast to the length of the \emph{encoding} of the problem).

These classes are accompanied by analogs of fpt reductions. These reduction schemes have slight technical differences to fpt reductions (\emph{q.v.} \cite{Weyer04}, \cite{FlumGroheWeyer06} and~\cite{ImpagliazzoPaturiZane01},  or \cite{FlumGrohe06} for a collected survey of these and other related work), however they still produce hierarchies akin to the $W$-hierarchy, for $t \geq 2$:

\[\EXPW[t] = \left[\WSAT{$\Gamma_{t,1}$}\right]^{\EXPT{}}\]

\[\EW[t] = \left[\WSAT{$\Gamma_{t,1}$}\right]^{\EPT{}}\]

Although the first levels of these hierarchies are more technically delicate than the $W$-hierarchy, we still have the following key identities:

\[\EXPW[1] = \left[\WSAT{$\Gamma_{1,2}^{-}$}\right]^{\EXPT{}}\]

and

\[\EW[1] = \left[\WSAT{$\Gamma_{1,2}^{-}$}\right]^{\EPT{}}\]

The hierarchy corresponding to \SUBEPT{} is mildly different\footnote{Incidentally \SUBEPT{} and the $S$-hierarchy correspond to parameterizations of the Exponential Time Hypothesis, making them particularly interesting parameterized classes. In fact, the entire $S$-hierarchy is contained in \EPT{}, with \EPT{} and \SUBEPT{} bearing a similar relationship as \XP{} and \FPT{}.}:

\[\SW[t] = \bigcup_{d\geq 1}\left[\text{\textsc{SAT($\Gamma_{t,d}$)}}\right]^{serf}\]

However we fortunately we also have that:

\[\SW[1] = \left[\text{$s$-$var$-}\WSAT{$\Gamma_{1,2}$}\right]^{serf}\]

Where \textsc{$s$-$var$-WSAT} is a different parameterization of the weighted satisfiability problem:

\pcproblem{$s$-$var$-WSAT($\Phi$)}{A formula $\phi \in \Phi$, an integer $k$.}{$\var(\phi)$ (the number of variables in $\phi$).}{Does $\phi$ have a satisfying assignment where $k$ variables are set to \textsc{True}?}

\subsection{The Miniaturization Isomorphism and the $M$-Hierarchy}

The $S$-hierarchy, despite being a bounded hierarchy of parameterized classes, reflects structure in the unbounded theory. This structure can be elucidated via the \emph{miniaturization isomorphism}. Given a parameterized problem $(\Pi,\kappa)$ over $\Sigma^{*}$ the \emph{miniaturization} of the problem is

\pcproblem{Mini-$(\Pi,\kappa)$}{$x\in\Sigma^{*}$, and $m\in\mathbb{N}$ in unary such that $\Card{x}\leq m$.}{$\lceil\frac{\kappa(x)}{\log m}\rceil$.}{Decide whether $x \in \Pi$.}

Under this mapping we have the following:
\[
(\Pi,\kappa) \in \text{\SUBEPT{}} \Leftrightarrow \text{\textsc{Mini-}}(\Pi,\kappa) \in \text{\FPT{}}
\]
Consequently we can define an intractability hierarchy via this relationship, the $M$-hierarchy. For the purposes of this paper we need only the following:
\[
(\Pi,\kappa) \in \SW[t]\text{-complete} \Leftrightarrow \text{\textsc{Mini-}}(\Pi,\kappa) \in \M[t]\text{-complete}
\]
However the $M$-hierarchy is closed under normal fpt reductions. The proofs of these results, and much more technical detail can be found in \cite{FlumGrohe06}, or the original papers~\cite{AbrahamsonDowneyFellows95, ChenCFHJKX04, ChenF04, ChenG07, ChenHKX04, DowneyEFPR03, FlumG04}, for context however, it is known that for all $t \geq 1$ we have $\M[t] \subseteq \W[t] \subseteq \M[t+1]$.

\section{Proof Checking, Interactive Proofs and PCPs}

\subsection{Notation and Notes}

For convenience we denote by $\mathbb{B}$ the set $\{0,1\}$.

The proof systems will often be phrased somewhat like interactive proofs, as this often seems an intuitive, natural presentation, however the proof string is in effect a table of polynomial coefficients indexed by length $m$ vectors over a field $\mathcal{F}$, along with the values of a truth assignment at points over this space.

\subsection{Basic Definitions}

\begin{definition}[PCP]
A \emph{Probabilistically Checkable Proof System} (a PCP) for a problem $\Pi$ over alphabet $\Sigma$ is a probabilistic polynomial-time Turing Machine $V$ that given input $x$ and access to a proof string $\sigma \in \Sigma^{*}$ satisfies the following conditions:
\begin{enumerate}
\item If $x$ is a \Yes{}-instance of $\Pi$, there is a $\sigma$ such that $V^{\sigma}$ accepts $x$ with probability $1$.
\item If $x$ is a \No{}-instance of $\Pi$, for every $\sigma$ the probability that $V^{\sigma}$ accepts $x$ is at most $\frac{1}{2}$.
\end{enumerate}
\end{definition}

The choice of $1$ and $\frac{1}{2}$ as the probabilities for the completeness and soundness of the verifier are in a sense somewhat arbitrary, for example, Babai, Fortnow \& Lund~\cite{BabaiFL91} use probabilities that vary with the length of the input, however the majority of results are stated directly with these probabilities, or are otherwise compatible.

\begin{definition}[Restricted PCP]
Given two functions $r, p: \mathbb{N} \rightarrow \mathbb{N}$, a PCP is \emph{$(r,p)$-restricted} if for every input $x$, $V$ uses at most $O(r(\Card{x}))$ random bits and $O(p(\Card{x}))$ bits of the proof string $\sigma$.
\end{definition} 

The set of all problems with a $(r,p)$-restricted PCP is typically denoted \pcp{r}{p}. With this notation we can thus succinctly restate Theorem~\ref{thm:pcp}:

\begin{theorem}[PCP Theorem~\cite{AroraLMSS98, AroraS98}]
\NP{} $=$ \pcp{\log n}{1}.
\end{theorem}

\subsection{Arithmetization Protocols}

Lund~\emph{et al.}~\cite{LundFKN92} introduced a protocol for demonstrating PCP and interactive proof results which they used to show that every problem in \Poly{}$^{\#\text{\Poly{}}}$ has an interactive proof (a key step in motivating Shamir's~\cite{Shamir90} result).

This protocol has proven to be extremely useful and has been used in whole or part for many of the PCP related results~\cite{AroraLMSS98, AroraS98, BabaiFL91, FeigeGLSS96, Shamir90}. It is worthwhile to sketch an outline of this protocol to give an intuition for the working of the main result of this paper.

Given a complexity class $\mathcal{C}$ we select a suitable $\mathcal{C}$-complete problem $\Pi$ and produce a verifier that completes the following tasks:
\begin{enumerate}
\item For input $x$, the verifier constructs an arithmetical representation $\phi$ of $x$ such that the value of $\phi$ is dependent on whether $x$ is a \Yes{}-instance of $\Pi$ or not. For example we may construct an arithmetic formula from a boolean formula such that the arithmetic formula is non-zero if and only if the boolean formula is satisfiable.
\item A sufficiently large field over which to do the arithmetic is chosen. Typically this will be $\mathbb{Z}_{p}$ for some sufficiently large prime $p$.
\item The verifier then checks the arithmetical representation a variable at a time by instantiating a single variable and obtaining a simplified representation in one variable from the proof which it can use to compare against the expected value. If the simplified representation is satisfactory, the verifier picks a random value from the field, permanently sets the variable to this value and replaces the expected value by the evaluation of the simplified expression with that random value.
\item Step 3 is repeated until some value does not match expectation, at which point the proof is rejected, or until all variables have been instantiated at which point the expression is checked explicitly using elements of the solution obtained from the proof (\emph{e.g.} values from a truth assignment).
\end{enumerate}

The key to the effectiveness of this protocol is in the restriction on the arithmetic representation and the size of the field. For clarity of discussion we will assume the representation to be a multinomial and the field to be $\mathbb{Z}_{p}$ for a sufficiently large prime $p$.

If the multinomial is of constant degree $d$, and the polynomial simplification over one variable obtained from the proof is false, it can agree with the true polynomial in at most $d$ places~\cite{Schwartz80}. So if the proof is false, it can ``look true'' for only a small number of values ($d$), and eventually some iteration of checking will observe an erroneous value with high probability ($1-\frac{dr}{p}$ where $r$ is the number of iterations). 

\section{Parameterized PCPs}

Clearly we can adapt PCP notions to parameterized complexity.

\begin{definition}[Parameterized PCP]
A \emph{Parameterized Probabilistically Checkable Proof System} (parameterized PCP, or \emph{p}-PCP) for parameterized problem $\Pi$ over alphabet $\Sigma$ is a probabilistic \FPT{}-time Turing Machine $V$ that given input $(x,k)$, an instance of $\Pi$, and access to an proof string $\sigma \in \Sigma^{*}$ satisfies the following conditions:
\begin{enumerate}
\item If $(x,k)$ is a \Yes{}-instance of $\Pi$, there is a $\sigma$ such that $V^{\sigma}$ accepts $(x,k)$ with probability $1$.
\item If $(x,k)$ is a \No{}-instance of $\Pi$, for any choice of $\sigma$ the probability that $V^{\sigma}$ accepts $(x,k)$ is no greater than $\frac{1}{2}$.
\end{enumerate}
\end{definition}

As with non-parameterized PCPs, the completeness and soundness probabilities need not be $1$ and $\frac{1}{2}$, however these values are sufficient for our purposes and confusing the notation thus serves no purpose.

\begin{definition}[Restricted \emph{p}-PCP]
Given two functions $r,p: \mathbb{N}\times\mathbb{N} \rightarrow \mathbb{N}$ a \emph{p}-PCP is \emph{$(r,p)$-restricted} if for every input $(x,k)$ it uses $O(r(\Card{x}, k))$ random bits and at most $O(p(\Card{x}, k))$ bits of the proof string $\sigma$. 
\end{definition}

We denote the set of all problems with an $(r,p)$-restricted \emph{p}-PCP by \ppcp{r}{p}.

For certain extreme values of the parameters, we can use the \ppcp{r}{p} notation to express some of the parameterized classes.

\begin{itemize}
\item \FPT{} $=$ \ppcp{0}{0}, by definition problems in $\FPT{}$ have no access to a proof and need no randomness.
\item \FPT{} $=$ \ppcp{f(k)+\log n}{0}. An $\FPT{}$-time algorithm can try all possible $f(k)+\log n$ random strings.
\item \FPT{} $=$ \ppcp{0}{f(k)+\log n}. An $\FPT{}$-time algorithm can generate all proofs of length $f(k)+\log n$.
\item \pNP{} $=$ \ppcp{0}{f(k)n^{O(1)}}. By definition.
\end{itemize}

\subsection{A Non-trivial Parameterized PCP for W[1]}

\begin{theorem}\label{thm:wsat-pcp}
\sloppypar Let $(\phi, k)$ be an instance of \WSATTN{} where $\max\{\var(\phi),\cl(\phi)\} \leq 2^{m}$. There is an $(m\log m, m\log m)$-restricted probabilistic $\FPT{}$-time Turing Machine that rejects $(\phi, k)$ with high probability if $(\phi, k)$ is a \No{}-instance of \WSATTN{}. That is, \WSATTN{} $\in$ \ppcp{m\log m}{m \log m}.
\end{theorem}

\begin{proof}
The protocol will follow the same general format as those of Lund \emph{et al.}~\cite{LundFKN92}, Babai, Fortnow \& Lund~\cite{BabaiFL91} and particularly Feige \emph{et al.}~\cite{FeigeGLSS96} in that we will construct an arithmetic representation of $\phi$ and use the proof to evaluate this function pointwise.

Let $\phi$ be a \textsc{2-CNF$^{-}$} with smallest $m$ such that $2^{m} \geq\{\var(\phi), \cl(\phi)\}$. Denote each clause and variable by a binary string over $m$ bits.

For $v \in \mathbb{B}^{m}$ and $i \in \{1,2\}$ define a set of functions $C_{c,i}: \mathbb{B}^{m} \rightarrow \mathbb{B}$ as
\[
C_{c,i}(v) = \left\{ \begin{array}{cl} 1 & \text{if $v$ is the $i^{th}$ variable of clause $c$} \\ 0 & \text{otherwise}\end{array}\right. 
\]
This can be done in such a fashion that each $C_{c,i}$ is multilinear in $m$ variables. We sketch an example; say that $v = v_{1}v_{2}v_{3} = 101$ is the $1^{st}$ variable of clause $c$, then $C_{c,1} = v_{1}(1-v_{2})v_{3}$. Then the only place (over $\mathbb{B}^{3}$) where this is $1$ is at $101$.

Let $A:\mathbb{B}^{m} \rightarrow \mathbb{B}$ be a truth assignment to the variables of $\phi$.

We then define the following function over some sufficiently large field.
\[
SC(A,y) = \sum_{x_{1},x_{2}\in \mathbb{B}^{m}} \prod_{i \in \{1,2\}} C_{y,i}(x_{i})A(x_{i})
\]
This evaluates to $0$ if and only if $A$ is a satisfying assignment for clause $y$.
Then $\phi$ in its entirety, can be expressed as:
\[
S(A) = \sum_{z\in\mathbb{B}^{m}} SC(A,z)\cdot\prod_{i \in [1,m]} r_{r}^{z_{i}}
\]
Where $z_{i}$ is the $i^{th}$ bit of the binary representation of $z$ and $(r_{1},\ldots,r_{m})$ is a set of independently chosen random numbers from $\mathcal{F}$. This additional term is included to ensure with high probability that in the extended function the sum is zero only when all clauses evaluate to zero under $A$ (again, Feige \emph{et al.}~\cite{FeigeGLSS96} demonstrate the correctness of this method).
However we must also verify that:
\[
\sum_{z \in \mathbb{B}^{m}} A(z) = k
\]
The first function now evaluates to zero if and only if all the clauses are satisfied and the second evaluates to $k$ if and only if the weight of the truth assignment is $k$.

We now employ the following proposition:
\begin{proposition}[\cite{BabaiFL91}, \cite{FeigeGLSS96}]
Given a field $\mathcal{F}$, every boolean function $f$ has a unique multilinear extension over $\mathcal{F}$. Moreover the value the extension at any point can be computed in time $2^{\arity(f)}$.
\end{proposition}

In particular we can compute the multilinear extension of $C$ in any field of our choosing. Then assuming that $A$ is close to multilinear, $S$ is a multinomial of constant degree. Of course we cannot simply compute $A$ in $\FPT{}$-time, otherwise we'd have no reason for a $p$-PCP! However Babai, Fortnow \& Lund~\cite{BabaiFL91} demonstrate a procedure for testing multilinearity of a function that fails with high probability if the function is not multilinear and succeeds otherwise. Feige \emph{et al.}~\cite{FeigeGLSS96} improve this test, reducing the number of random and proof bits required to $O(m\log m)$.

We may now apply a protocol in the style of Lund \emph{et al.}~\cite{LundFKN92}, though Feige \emph{et al.}'s~\cite{FeigeGLSS96} version of the protocol is the direct inspiration.

Given a multinomial $h$ of constant degree $d$ over $q$ variables the function $g_{i}(x_{i})$ where the first $i-1$ variables are randomly instantiated
\[
g_{i}(x_{i}) = \sum_{x_{i+1}, \ldots, x_{q} \in \mathbb{B}} h(r_{1}, \ldots, r_{i-1}, x_{i}, \dots, x_{q})
\]
is a polynomial of degree $d$.

Assuming $A$ is multilinear with high probability (to ensure the degree bound of the multinomial), given an expected value $a_{i-1}$ we perform the $i^{th}$ iteration of the proof check as follows:

\begin{enumerate}
\item Obtain from the proof the $d$ coefficients of the polynomial $g_{i}'$ that is purported to be $g_{i}$.
\item Check that $g_{i}'(0)+g_{i}'(1) = a_{i-1}$, if not, then reject.
\item If the first check passes, we may still have $g_{i} \neq g_{i}'$. However they can agree at at most $d$ points in $\mathcal{F}$. We can check this with high probability ($1-\frac{d}{\Card{\mathcal{F}}}$) by randomly picking a value $r_{i}$, setting $a_{i} := g_{i}'(r_{i})$ and verifying the formula recursively.
\end{enumerate}

Initially we have $a_{0} = 0$. The process continues until all variables have been randomly instantiated, at which point we can check the final function directly by obtaining the two values of $A$ at the randomly generated points described by the instantiated variables and computing the value. By choosing $\mathcal{F}$ such that $\Card{\mathcal{F}}> \frac{md}{\varepsilon}$, the probability of accepting at some point over the $m$ rounds is $\varepsilon$.

The function checking the weight of the satisfying assignment can be checked using the same protocol.

As $\log\Card{\mathcal{F}} \in O(\log m)$, this protocol uses $O(m\log m)$ proof bits to obtain the polynomial coefficients and $O(m\log m)$ random bits in instantiating the function.

\end{proof}

\begin{corollary}\label{cor:w[1]-pcp}
For every parameterized problem $\Pi \in \W[1]$ there exists a function $f: \mathbb{N} \rightarrow \mathbb{N}$ such that $\Pi \in$ \ppcp{(f(k)+\log n) \log(f(k)+\log n)}{(f(k)+\log n) \log(f(k)+\log n)} and hence $\W[1] \subseteq$ \ppcp{(f(k)+\log n) \log(f(k)+\log n)}{(f(k)+\log n) \log(f(k)+\log n)} where $n$ is the size of the instance and $k$ is the parameter.
\end{corollary}

\begin{proof}
As \WSATTN{} is $\W[1]$-complete, every problem in $\W[1]$ can be reduced to an instance of \WSATTN{} in time bounded by $f(k)n^{O(1)}$ for some computable function $f$. Hence the instance of \WSATTN{} produced by the reduced has at most $f(k)n^{O(1)}$ variables and $f(k)n^{O(1)}$ clauses.
\end{proof}

\subsection{Unbounded Clauses and $\W[2]$}

The class $\Gamma^{+}_{2,1}$ of propositional formul\ae{} can be more naturally thought of as the class of all propositional CNF formul\ae{}. The protocol given for $\W[1]$ in the previous section, although defined for $\Gamma^{-}_{1,2}$, does not depend on the clause length --- the bounds on the number of bits used may change, but the clause length is not fundamental to the structure, unlike say, that the formula is in CNF as this restriction ensures that the arithmetization is multilinear.

\begin{theorem}
\sloppypar Let $(\phi, k)$ be an instance of \WSAT{$\Gamma^{+}_{2,1}$} where $\max\{\var(\phi),\cl(\phi)\} \leq 2^{m}$ and $p$ is the length of the longest clause. There is an $(p\cdot m\log m, p\cdot m\log m)$-restricted probabilistic $\FPT{}$-time Turing Machine that rejects $(\phi, k)$ with high probability if $(\phi, k)$ is a \No{}-instance of \WSAT{$\Gamma^{+}_{2,1}$}. That is, \WSAT{$\Gamma^{+}_{2,1}$} $\in$ \ppcp{p \cdot m\log m}{p \cdot m \log m}.
\end{theorem}

\begin{proof}
We can modify the $SC$ function to cope with greater clause length and positive rather than negative literals:
\[
SC(A,y) = \sum_{x_{1},\ldots,x_{p}\in \mathbb{B}^{m}} \prod_{i \in \{1,p\}} C_{y,i}(x_{i})(1-A(x_{i}))
\]
The family of functions $C_{c,i}$ is also extended in the obvious way.

Then the protocol continues for $p\cdot m$ rounds rather than the $2\cdot m$ as for the $\Gamma^{-}_{1,2}$ case. We then need a factor of $p$ extra random bits, and we require $p$ values of the satisfying assignment $A$ for the final evaluation.
\end{proof}

\begin{corollary}
For every parameterized problem $\Pi \in \W[2]$ there exists a function $f: \mathbb{N} \rightarrow \mathbb{N}$ such that $\Pi \in$ \ppcp{p\cdot(f(k)+\log n) \log(f(k)+\log n)}{p\cdot(f(k)+\log n) \log(f(k)+\log n)} and hence $\W[2] \subseteq$ \ppcp{p\cdot(f(k)+\log n) \log(f(k)+\log n)}{p\cdot(f(k)+\log n) \log(f(k)+\log n)} where $n$ is the size of the instance and $k$ is the parameter and $p$ is the length of the longest clause in the equivalent \WSAT{$\Gamma^{+}_{2,1}$} instance.
\end{corollary}

The catch with this of course is that $p$ may, in principle, be as long as the formula and hence $O(f(k)n^{O(1)})$, in which case we do no better (actually, clearly worse) than the trivial \emph{p}-PCP guaranteed by the fact that $\W[2]\subseteq \text{\pNP{}}$.

\subsection{Extension to Bounded Parameterized Classes}


As \WSATTN{} is complete for both $\EXPW[1]$~\cite{Weyer04} and $\EW[1]$~\cite{FlumGroheWeyer06}, we can easily adapt the $\W[1]$ result. We omit the formal particulars of the restriction on the running time and reduction structures denoting them simply by prepending the bound to the nomenclature.

\begin{corollary}
$\EXPW[1] \subseteq$ $2^{k^{O(1)}}$-\ppcp{(2^{k^{O(1)}}+\log n) \log(2^{k^{O(1)}}+\log n)}{(2^{k^{O(1)}}+\log n) \log(2^{k^{O(1)}}+\log n)} where $n$ is the size of the instance and $k$ is the parameter.
\end{corollary}

\begin{corollary}
$\EW[1] \subseteq$ $2^{O(k)}$-\ppcp{(2^{O(k)}+\log n) \log(2^{O(k)}+\log n)}{(2^{O(k)}+\log n) \log(2^{O(k)}+\log n)} where $n$ is the size of the instance and $k$ is the parameter.
\end{corollary}

As \WSATTN{} is not $\SW[1]$-complete, we need to adjust the formula $CS$ used in Theorem~\ref{thm:pcp} as we can no longer assume that all variables are negated. Fortunately we can simply use the formula of Feige \emph{et al.}~\cite{FeigeGLSS96} more directly (adjusted for 2-CNF, rather than 3-CNF). Recall that $2^{m} \geq \max\{\var(\phi),\cl(\phi)\}$ (of course we are really just interested in taking a power of two so that the logarithms work neatly). As $\var(\phi)=k'$ is the parameter we know that $m \leq \log(4k'^{2})$. Given that $k$ is the parameter of the initial problem and $n$ is the size, the reduction scheme that closes the ${\normalfont S}$-hierarchy gives $k' = g(l)(k+\log n)$ for some $\SUBEPT{}$-time computable function $g$ over $\mathbb{N}$. 

\begin{corollary}\sloppypar
$\SW[1] \subseteq$ $2^{o^{eff}(k)}$-\ppcp{\log (g'(l)(k + \log n)^{2}) \log \log (g'(l)(k + \log n)^{2})}{\log (g'(l)(k + \log n)^{2}) \log \log( g'(l)(k + \log n)^{2})} where $n$ is the size of the instance, $k$ is the parameter and $g'$ is a $\SUBEPT{}$-time computable function over $\mathbb{N}$.
\end{corollary}

Then from the miniaturization isomorphism we get:

\begin{corollary}\sloppypar
$\M[1] \subseteq$ \ppcp{\log (f(\frac{k}{\log n})n^{O(1)}) \log \log (f(\frac{k}{\log n})n^{O(1)})}{\log (f(\frac{k}{\log n})n^{O(1)}) \log \log (f(\frac{k}{\log n})n^{O(1)})} where $n$ is the size of the instance and $\frac{k}{\log n}$ is the parameter.
\end{corollary}

\section{Proof Checking for the $A$-Hierarchy}

Looking at the classes of the $A$-hierarchy, one may be put in mind of Shamir's~\cite{Shamir90} proof that IP=\PSPACE{} via a Lund \emph{et al.}~\cite{LundFKN92} style protocol over instances of the \textsc{Quantified Boolean Satisfiability} problem. However, the restriction of the weight of the solution poses some interesting problems. While in Shamir's case, the universal quantification is truly universal, in ours it is universal only in the ``for all subsets of size $k$'' sense, hence it is difficult to translate an instance of \textsc{AWSAT} in the same fashion --- dealing with each universally quantified variable individually becomes complicated by the fact that its possible values depend on how many of the previous variables have been set to \textsc{TRUE}, which is further complicated by the assignment of a random value out of a much larger field.

From the parameterized perspective, it is also perhaps not sensible that we ask to verify a membership proof of an \textsc{AWSAT} problem in \FPT{}-time. If we consider a certificate for such an instance (with the technical consideration that $l\geq 2$) then we must verify not only a single weight $k$ satisfying assignment, but for those variables that are universally quantified, we must verify that \emph{all} weight $k_{i}$ assignments have accompanying assignments from the existentially quantified variables following them. That is, we are in effect expected to check on the order of $n^{k}$ assignments. This is reflected in the structure of the parameterized classes --- while the $W$-hierarchy is contained in \pNP{}, apart from $\A[1]$, there is no evidence that the $A$-hierarchy is. However the $A$-hierachy is contained in \XP{}, hence we can solve these problems in time $f(k)+n^{f(k)}$ and naturally can thus check solutions within that bound.

With this in mind we suggest a slightly relaxed version of a parameterized PCP, where we make the obvious changes from \FPT{}-time to $f(k)+n^{f(k)}$. For simplicity we will denote this as a $n^{k}$-$p$-PCP.

\begin{theorem}
\sloppypar Let $(\phi, X_{1},\ldots, X_{l}, k=k_{1}+\ldots+k_{l})$ be an instance of \AWSATL{l}{$\Gamma^{-}_{1,2}$} where $\max\{\var(\phi),\cl(\phi)\} \leq 2^{m}$ with $l$ odd and $k' = k_{2}+k_{4}+\ldots+k_{l-1}$. There is an $(n^{k'}\cdot m\log m, n^{k'}\cdot m\log m)$-restricted probabilistic $(f(k)+n^{f(k)})$-time Turing Machine that rejects $(\phi, k)$ with high probability if $(\phi, X_{1},\ldots, X_{l}, k=k_{1}+\ldots+k_{l})$ is a \No{}-instance of \AWSATL{l}{$\Gamma^{-}_{1,2}$}. That is, \AWSATL{l}{$\Gamma^{-}_{1,2}$} $\in$ \ppcp{n^{k'}\cdot m\log m}{n^{k'}\cdot m\log m}.
\end{theorem}

\begin{proof}
The verifying TM $V$ begins by generating the $O(n^{k'})$ assignments to the variables of $X_{even} = X_{2} \cup X_{4} \cup \ldots X_{l-1}$. In effect we can treat this as a simple string $s$ over $\{0,1\}^{\Card{X_{even}}}$, which we will use to index elements of the truth assignment given in the proof string (which again we can treat as a table). For each assignment to $X_{even}$ we can reduce the input formula $\phi$ appropriately in polynomial time, substituting in the values of the literals and simplifying the formula to $\phi'$.

We then have a series of $\Gamma_{1,2}^{-}$ formul\ae{} with only \emph{existential} qualification, but this is equivalent to an instance of \WSAT{$\Gamma_{1,2}^{-}$}, only with the slight constraint that the truth assignment is required to consist of $\frac{l+1}{2}$ parts, corresponding to the odd indexed variable sets $X_{1},\ldots,X_{l}$.

Thus we can apply the protocol used for $\W[1]$, with the slight change that instead of checking simply that $\sum_{x_{i}}A(x_{i})=k$, we check the sequence of truth assignments $A^{s}_{j}$ where $j \in \SB 2h-1 \SM h \in \mathbb{N}^{+} \SE$, ensuring that for each the weight is $k_{j}$. 
\end{proof}

\begin{corollary}
$\A[l] \subseteq$ \ppcp{f(k)n^{g(k)}\cdot \log (f(k)n^{O(1)})\log\log (f(k)n^{O(1)})}{f(k)n^{g(k)}\cdot \log (f(k)n^{O(1)})\log\log (f(k)n^{O(1)})} for all $l \geq 1$, where $n$ is the size of the input, $k$ is the parameter and $g$ and $f$ are computable functions.
\end{corollary}

\begin{proof}
As $\A[l]$ is closed under fpt-reductions, if $l$ is odd, we can reduce the input instance to an instance of \AWSATL{l}{$\Gamma^{-}_{1,2}$} with at most $f(k)n^{O(1)}$ clauses and variables, with parameter $g(k)$.

By containment, if $l$ is even, we can reduce the input to an instance of \AWSATL{l+1}{$\Gamma^{-}_{1,2}$}.
\end{proof}

We note particularly that this $p$-PCP has the nice property of reducing to the $\W[1]$ $p$-PCP in the case where $l=1$. This is a generally desirable property as $\A[1] = \W[1]$ (though in general we only expect that $\W[t] \subseteq \A[t]$).

\begin{corollary}
$\AW[*] \subseteq$ \ppcp{f(k)n^{g(k)\cdot \lfloor\frac{l}{2}\rfloor}\cdot \log (f(k)n^{O(1)})\log\log (f(k)n^{O(1)})}{f(k)n^{g(k)\cdot \lfloor\frac{l}{2}\rfloor}\cdot \log (f(k)n^{O(1)})\log\log (f(k)n^{O(1)})}
\end{corollary}

\begin{proof}
Any problem in $\AW[*]$ can be reduced to an instance of \AWSAT{$\Gamma^{-}_{1,2}$}. In this case $l$ is not fixed, but part of the input. However for a given instance, the number of even-index variable sets is at most $\lfloor\frac{l}{2}\rfloor$.
\end{proof}

\section{Conclusion}

The development of parameterized PCPs, of which this is simply a first step, may have interesting results, particularly for parameterized approximation theory. Currently non-trivial parameterized approximations are few, and the status of key problems such as \textsc{Clique} and \textsc{Dominating Set} are essentially unknown. For parameterized PCPs to have an impact on this however, results need to be improved and extended. By employing directly the construction of Feige~\emph{et al.}~\cite{FeigeGLSS96} for \textsc{Max-Clique} we could obtain results if we can reduce the number of random bits of a $p$-PCP containing $\W[1]$ to a function of $k$ alone. This seems possible for the main part of the checking protocol --- we can simply randomly generate only $k$ of the values, and take all others as constant (say 0), with a corresponding alteration in the size of the field over which the values are generated, the probability of incorrectly accepting is in essence no different. A similar alteration to the multilinearity testing however is much more difficult. Another possible approach would be to explore the intersection of Dinur's~\cite{Dinur07} proof of the PCP theorem which employs certain constraint satisfaction problems and recent hardness results for parameterized versions of constraint satisfaction~\cite{BulatovMarx11}.

Extending the result of this paper to cover other classes also seems to be non-trivial, the alternation of boolean operators of unbounded arity in propositional classes that define the classes $\W[t]$ seems to preclude retaining the constant degree property essential to the protocol presented here (this is not a problem for $\NP{}$ as we do not need to keep track of the weight of the satisfying assignment, so the polynomial expansion experienced in reducing a formula to \textsc{3-CNF} creates no problem). However it seems likely that a tight $p$-PCP for $\W[1]$ would be part of a broader $p$-PCP that generalizes to $\W[t]$ for all $t$, implying that $t$ will play an important role in the final complexity description.

In the other direction it would be interesting to obtain a more general $p$-PCP for the other \PSPACE{} related parameterized classes, particularly $\AW[SAT]$ and $\AW[P]$.

\bibliography{PCP}

\begin{thebibliography}{10}

\bibitem{AbrahamsonDowneyFellows95}
K.~A. Abrahamson, R.~G. Downey, and M.~R. Fellows.
\newblock Fixed-parameter tractability and completeness {IV}: {On} completeness
  for {W[P]} and {PSPACE} analogs.
\newblock {\em Annals of Pure and Applied Logic}, 73:235--276, 1995.

\bibitem{AroraLMSS98}
Sanjeev Arora, Carsten Lund, Rajeev Motwani, Madhu Sudan, and Mario Szegedy.
\newblock Proof verification and the hardness of approximation problems.
\newblock {\em Journal of the ACM}, 45(3):501--555, 1998.

\bibitem{AroraS92}
Sanjeev Arora and Shmuel Safra.
\newblock Probabilistic checking of proofs; a new characterization of np.
\newblock In {\em 33rd Annual Symposium on Foundations of Computer Science,
  Pittsburgh, Pennsylvania, USA, 24-27 October 1992}, pages 2--13. IEEE
  Computer Society, 1992.

\bibitem{AroraS98}
Sanjeev Arora and Shmuel Safra.
\newblock Probabilistic checking of proofs: A new characterization of {NP}.
\newblock {\em Journal of the ACM}, 45(1):70--122, 1998.

\bibitem{Babai85}
L{\'a}szl{\'o} Babai.
\newblock Trading group theory for randomness.
\newblock In Robert Sedgewick, editor, {\em Proceedings of the 17th Annual ACM
  Symposium on Theory of Computing, May 6-8, 1985, Providence, Rhode Island,
  USA}, pages 421--429. ACM, 1985.

\bibitem{BabaiFLS91}
L{\'a}szl{\'o} Babai, Lance Fortnow, Leonid~A. Levin, and Mario Szegedy.
\newblock Checking computations in polylogarithmic time.
\newblock In Cris Koutsougeras and Jeffrey~Scott Vitter, editors, {\em
  Proceedings of the 23rd Annual ACM Symposium on Theory of Computing, May 5-8,
  1991, New Orleans, Louisiana, USA}, pages 21--31. ACM, 1991.

\bibitem{BabaiFL91}
L{\'a}szl{\'o} Babai, Lance Fortnow, and Carsten Lund.
\newblock Non-deterministic exponential time has two-prover interactive
  protocols.
\newblock {\em Computational Complexity}, 1:3--40, 1991.

\bibitem{BabaiFL92}
L{\'a}szl{\'o} Babai, Lance Fortnow, and Carsten Lund.
\newblock Addendum to non-deterministic exponential time has two-prover
  interactive protocols.
\newblock {\em Computational Complexity}, 2:374, 1992.

\bibitem{BabaiM88}
L{\'a}szl{\'o} Babai and Shlomo Moran.
\newblock Arthur-merlin games: A randomized proof system, and a hierarchy of
  complexity classes.
\newblock {\em Journal of Computer and System Sciences}, 36(2):254--276, 1988.

\bibitem{Ben-OrGKW88}
Michael Ben-Or, Shafi Goldwasser, Joe Kilian, and Avi Wigderson.
\newblock Multi-prover interactive proofs: How to remove intractability
  assumptions.
\newblock In Janos Simon, editor, {\em Proceedings of the 20th Annual ACM
  Symposium on Theory of Computing, May 2-4, 1988, Chicago, Illinois, USA},
  pages 113--131. ACM, 1988.

\bibitem{BulatovMarx11}
Andrei~A. Bulatov and D{\'a}niel Marx.
\newblock Constraint satisfaction parameterized by solution size.
\newblock In Luca Aceto, Monika Henzinger, and Jiri Sgall, editors, {\em
  Automata, Languages and Programming - 38th International Colloquium, ICALP
  2011, Zurich, Switzerland, July 4-8, 2011, Proceedings, Part I}, volume 6755
  of {\em Lecture Notes in Computer Science}, pages 424--436. Springer, 2011.

\bibitem{ChenCFHJKX04}
Jianer Chen, Benny Chor, Mike Fellows, Xiuzhen Huang, David~W. Juedes, Iyad~A.
  Kanj, and Ge~Xia.
\newblock Tight lower bounds for certain parameterized np-hard problems.
\newblock In {\em 19th Annual IEEE Conference on Computational Complexity (CCC
  2004), 21-24 June 2004, Amherst, MA, USA}, pages 150--160. IEEE Computer
  Society, 2004.

\bibitem{ChenHKX04}
Jianer Chen, Xiuzhen Huang, Iyad~A. Kanj, and Ge~Xia.
\newblock Linear fpt reductions and computational lower bounds.
\newblock In L{\'a}szl{\'o} Babai, editor, {\em Proceedings of the 36th Annual
  ACM Symposium on Theory of Computing, Chicago, IL, USA, June 13-16, 2004},
  pages 212--221. ACM, 2004.

\bibitem{ChenF04}
Yijia Chen and J{\"o}rg Flum.
\newblock On miniaturized problems in parameterized complexity theory.
\newblock In Rodney~G. Downey, Michael~R. Fellows, and Frank K. H.~A. Dehne,
  editors, {\em Parameterized and Exact Computation, First International
  Workshop, IWPEC 2004, Bergen, Norway, September 14-17, 2004, Proceedings},
  Lecture Notes in Computer Science, pages 108--120. Springer, 2004.

\bibitem{ChenG07}
Yijia Chen and Martin Grohe.
\newblock An isomorphism between subexponential and parameterized complexity
  theory.
\newblock {\em SIAM Journal of Computing}, 37(4):1228--1258, 2007.

\bibitem{Condon91}
Anne Condon.
\newblock The complexity of the max word problem.
\newblock In Christian Choffrut and Matthias Jantzen, editors, {\em STACS 91,
  8th Annual Symposium on Theoretical Aspects of Computer Science, Hamburg,
  Germany, February 14-16, 1991, Proceedings}, volume 480 of {\em Lecture Notes
  in Computer Science}, pages 456--465. Springer, 1991.

\bibitem{Dinur07}
Irit Dinur.
\newblock The {PCP} theorem by gap amplification.
\newblock {\em Journal of the ACM}, 54(3):44, 2007.

\bibitem{DowneyEFPR03}
Rodney~G. Downey, Vladimir Estivill-Castro, Michael~R. Fellows, Elena Prieto,
  and Frances~A. Rosamond.
\newblock Cutting up is hard to do: the parameterized complexity of k-cut and
  related problems.
\newblock {\em Electronic Notes on Theoretical Computer Science}, 78:209--222,
  2003.

\bibitem{DowneyFellows99}
Rodney~G. Downey and Michael~R. Fellows.
\newblock {\em Parameterized Complexity}.
\newblock Springer, 1999.

\bibitem{FeigeGLSS96}
Uriel Feige, Shafi Goldwasser, L{\'a}szl{\'o} Lov{\'a}sz, Shmuel Safra, and
  Mario Szegedy.
\newblock Interactive proofs and the hardness of approximating cliques.
\newblock {\em Journal of the ACM}, 43(2):268--292, 1996.

\bibitem{FlumG04}
J{\"o}rg Flum and Martin Grohe.
\newblock Parametrized complexity and subexponential time (column:
  Computational complexity).
\newblock {\em Bulletin of the EATCS}, 84:71--100, 2004.

\bibitem{FlumGrohe06}
J{\"o}rg Flum and Martin Grohe.
\newblock {\em Parameterized complexity theory}.
\newblock Springer, 2006.

\bibitem{FlumGroheWeyer06}
J{\"o}rg Flum, Martin Grohe, and Mark Weyer.
\newblock Bounded fixed-parameter tractability and log$^{\mbox{2}}${\it n}
  nondeterministic bits.
\newblock {\em Journal of Computer and System Sciences}, 72(1):34--71, 2006.

\bibitem{GoldwasserMR85}
Shafi Goldwasser, Silvio Micali, and Charles Rackoff.
\newblock The knowledge complexity of interactive proof-systems (extended
  abstract).
\newblock In Robert Sedgewick, editor, {\em Proceedings of the 17th Annual ACM
  Symposium on Theory of Computing, May 6-8, 1985, Providence, Rhode Island,
  USA}, pages 291--304. ACM, 1985.

\bibitem{GoldwasserMR89}
Shafi Goldwasser, Silvio Micali, and Charles Rackoff.
\newblock The knowledge complexity of interactive proof systems.
\newblock {\em SIAM Journal of Computing}, 18(1):186--208, 1989.

\bibitem{GoldwasserS86}
Shafi Goldwasser and Michael Sipser.
\newblock Private coins versus public coins in interactive proof systems.
\newblock In Juris Hartmanis, editor, {\em Proceedings of the 18th Annual ACM
  Symposium on Theory of Computing, May 28-30, 1986, Berkeley, California,
  USA}, pages 59--68. ACM, 1986.

\bibitem{ImpagliazzoPaturiZane01}
Russell Impagliazzo, Ramamohan Paturi, and Francis Zane.
\newblock Which problems have strongly exponential complexity?
\newblock {\em Journal of Computer and System Sciences}, 63(4):512--530, 2001.

\bibitem{LundFKN92}
Carsten Lund, Lance Fortnow, Howard~J. Karloff, and Noam Nisan.
\newblock Algebraic methods for interactive proof systems.
\newblock {\em Journal of the ACM}, 39(4):859--868, 1992.

\bibitem{Schwartz80}
J.~T. Schwartz.
\newblock Fast probabilistic algorithms for verification of polynomial
  identities.
\newblock {\em Journal of the ACM}, 27(4):701--717, 1980.

\bibitem{Shamir90}
Adi Shamir.
\newblock {IP=PSPACE}.
\newblock In {\em 31st Annual Symposium on Foundations of Computer Science, St.
  Louis, Missouri, USA, October 22-24, 1990, Volume I}, pages 11--15. IEEE
  Computer Society, 1990.

\bibitem{Weyer04}
Mark Weyer.
\newblock Bounded fixed-parameter tractability: The case 2poly(k).
\newblock In Rodney~G. Downey, Michael~R. Fellows, and Frank K. H.~A. Dehne,
  editors, {\em Parameterized and Exact Computation, First International
  Workshop, IWPEC 2004, Bergen, Norway, September 14-17, 2004, Proceedings},
  volume 3162 of {\em Lecture Notes in Computer Science}, pages 49--60.
  Springer, 2004.

\end{thebibliography}
\bibliographystyle{plain}

\end{document}